\documentclass[acmsmall,nonacm]{acmart}

\usepackage[utf8]{inputenc}
\usepackage{xspace}
\usepackage{amsmath}
\usepackage{amsthm}
\usepackage{amsfonts}
\usepackage{MnSymbol}
\usepackage{mathtools}
\usepackage{textcomp}
\usepackage{varwidth}
\usepackage{graphicx}
\usepackage{mathrsfs} 
\usepackage{subfiles}
\usepackage{subcaption}

\usepackage{tikz}
\usetikzlibrary{automata, graphs,positioning,chains,arrows}

\newtheorem{theorem}{Theorem}
\newtheorem{proposition}[theorem]{Proposition}
\newtheorem{lemma}[theorem]{Lemma}
\newtheorem{corollary}[theorem]{Corollary}
\newtheorem{observation}[theorem]{Observation}

\theoremstyle{definition}
\newtheorem{definition}[theorem]{Definition}
\newtheorem{example}[theorem]{Example}

\AtBeginEnvironment{example}{%
    \pushQED{\qed}%
}
\AtEndEnvironment{example}{\popQED\endexample}

\usepackage{todonotes}

\newcommand{\Int}{\mathsf{Int}}
\newcommand{\E}{\mathop{\mathbb{E}}}

\newcommand{\CFG}{\textsf{CFG}}
\newcommand{\uCFG}{\textsf{uCFG}}

\newcommand{\Lset}{\mathcal{L}}

\renewcommand{\paragraph}[1]{\smallskip
	\noindent\textbf{#1}.}

\renewcommand{\setminus}{\,\,\backslash\,\,}

\title[A Lower Bound on Unambiguous Context Free Grammars]{A Lower Bound on Unambiguous Context Free Grammars via Communication Complexity}

\author{Stefan Mengel}
\orcid{0000-0003-1386-8784}
\affiliation{		
    \institution{Univ. Artois, CNRS, Centre de Recherche en Informatique de Lens (CRIL)}
    \city{Lens}
    \country{France}}
\email{mengel@cril.fr}

\author{Harry Vinall-Smeeth}
\orcid{0000-0003-2422-9435}
\affiliation{
    \institution{Technische Universit\"at Ilmenau}
    \city{Ilmenau}
    \country{Germany}
}
\email{harry.vinall-smeeth@tu-ilmenau.de}



\begin{document}

\begin{abstract}
Motivated by recent connections to factorised databases, we analyse the efficiency of representations by context free grammars ({\CFG}s). Concretely, we prove a recent conjecture by Kimelfeld, Martens, and Niewerth (ICDT 2025), that for finite languages representations by general {\CFG}s can be doubly-exponentially smaller than those by unambiguous {\CFG}s. To do so, we show the first exponential lower bounds for representation by unambiguous {\CFG}s of a finite language that can efficiently be represented by ambiguous {\CFG}s. Our proof first reduces the problem to proving a lower bound in a non-standard model of communication complexity. Then, we argue similarly in spirit to a recent discrepancy argument to show the required communication complexity lower bound. Our result also implies that a finite language may admit an exponentially smaller representation as a nondeterministic finite automaton than as an unambiguous {\CFG}. 
\end{abstract} 	
	\maketitle

	\section{Introduction}\label{sec:introduction}

Over recent years, the succinct representation of query results has become a  prominent topic within database theory. This raises a fundamental question: how do different representation formats compare to one another? In this paper, we compare context free grammars ({\CFG}) to their unambiguous variant ({\uCFG}). In a nutshell, we show that enforcing unambiguity---which is useful for applications like counting and enumeration---comes at a heavy price in terms of succinctness. 

At this point some readers may wonder what {\CFG}s have to do with representing query answers. In fact, recent work by Kimelfeld, Martens and Niewerth \cite{DBLP:journals/corr/abs-2309-11663} observed a close connection between {\CFG}s and 
\emph{d-representations}---a form of factorised representations introduced over a decade ago by Olteanu and Z\'{a}vodn\'{y} in a landmark paper \cite{DBLP:journals/tods/OlteanuZ15}. To be precise, {\CFG}s representing finite languages are isomorphic to d-representations in the \emph{unnamed perspective}, see~\cite{DBLP:journals/corr/abs-2309-11663} for details. 

Let us highlight two nice features of d-representations. Firstly, given a query and a database we can often directly compute a d-representation which moreover, may be exponentially smaller than the materialisation of all query answers \cite{DBLP:journals/tods/OlteanuZ15}. And secondly, we can carry out many algorithms directly on d-representations \cite{DBLP:journals/pvldb/BakibayevKOZ13}. These representations, and close variants thereof, have been used in various settings e.g. databases under updates \cite{DBLP:conf/pods/BerkholzKS17, DBLP:conf/pods/Olteanu24}, representing provenance \cite{DBLP:conf/icdt/OlteanuZ12} and machine learning \cite{DBLP:conf/sigmod/SchleichOC16}.
More broadly, d-representations can be placed within the framework of circuits originating in
 \emph{knowledge compilation}, a sub-field of artificial intelligence concerning the efficient representation of various forms of knowledge and data~\cite{DBLP:journals/jair/DarwicheM02, DBLP:conf/aaai/Olteanu16}. Techniques from knowledge compilation have been widely deployed in database theory in recent years, see \cite{DBLP:journals/sigmod/AmarilliC24} for a survey. 
 Moreover, the circuits corresponding to {\uCFG}s have been used in the context of enumeration \cite{DBLP:conf/icalp/AmarilliBJM17}.

From an algorithmic perspective, {\uCFG}s offer several advantages over {\CFG}s. For example, it is well know that while counting is in polynomial time for {\uCFG}s, it is \#P-complete for {\CFG}s. Similarly, {\uCFG}s allow for more efficient enumeration, e.g. via \cite{DBLP:conf/icalp/AmarilliBJM17, DBLP:conf/icdt/MunozR22}. Different variants of {\CFG}s have also been used in information extraction~\cite{DBLP:conf/icdt/Peterfreund21,DBLP:conf/pods/AmarilliJMR22} and unambiguity is also crucual there. But are there finite languages which are accepted by a small {\CFG} and not by any small {\uCFG}? This is a question left open by Kimelfeld, Martens and Niewerth \cite{DBLP:journals/corr/abs-2309-11663}.
They conjecture a double exponential separation for {\CFG}s for a concrete language. However, they say `[t]o the best of our knowledge, the literature does not yet have
well-developed methods for proving size lower bounds for {\uCFG}s.' In this paper, we develop such
 methods and prove the conjecture from~\cite{DBLP:journals/corr/abs-2309-11663} via the following theorem. 

\begin{theorem} \label{thm:main}
For every $n\in \mathbb{N}$, there is finite language $L_n$ over a binary alphabet in which all words have length $2n$ such that:
\begin{enumerate}
\item $L_n$ is accepted by a {\CFG} of size $\Theta(\log(n))$,
\item $L_n$ is accepted by a non-deterministic finite automaton of size $\Theta(n)$ and
\item every ${\uCFG}$ accepting $L_n$ has size $2^{\Omega(n)}$.
\end{enumerate} 
\end{theorem}

Moreover, we may take $L_n$ to be exactly the language from the conjecture of Kimelfeld, Martens and Niewerth~\cite{DBLP:journals/corr/abs-2309-11663}. This language is very natural: it consists exactly of those strings of length $2n$ over the alphabet $\{a,b\}$ such that there are two $a$ symbols at distance exactly $n$ from one-another.

Conceptually, our techniques not only answer the call for size lower bound methods for {\uCFG}s but contribute more generally to our understanding of the role of unambiguity within computation.
This is an area which has received quite some attention within the formal language community. For example, unambiguity within automata theory was the subject of a recent Dagstuhl seminar \cite{DBLP:journals/dagstuhl-reports/ColcombetQS21}, see \cite{DBLP:conf/dcfs/Colcombet15} for a survey on this topic. It is not only in the context of {\CFG}s that unambiguity has proved difficult to deal with. One reason for this is probably that unambiguity is a semantic property, not a syntactic one, making it hard to `get your hands on' in a proof. 

As a result, basic questions about unambiguity remain open and others have only recently been resolved. For example, whether unambiguous finite automata admit polynomial time complementation was only relatively recently answered in the negative \cite{DBLP:conf/icalp/Raskin18,DBLP:conf/icalp/GoosK022}. Similarly, within knowledge compilation, whether d-DNNF circuits admit polynomial time negation has been open for over twenty years and even the special case of structured d-DNNF has only recently been resolved \cite{DBLP:conf/ijcai/Vinall-Smeeth24}.

\paragraph{On the Naturalness of $L_n$} Let us quickly comment on the naturalness of the language $L_n$ that we consider. Of course, the lower bound on $L_n$ is not directly a result in a practical setting. However, it encapsulates in a prototypical fashion a problem of unambiguous representations: how can one unambiguously deal with unions of highly non-disjoint sets? This is the underlying difficulty for many problems, e.g. the hardness of counting for NFAs and CFGs, and here we show that it is also hard for a type of succinct representation.

An advantage of studying $L_n$ as a very bare-bones version of the problem is that one can hopefully then use it to show hardness of other, more practical problems. This is the reason why different versions of Set Disjointness (essentially a complement of $L_n$) have been studied extensively in the communication complexity literature and then used in many lower bounds, see e.g. the applications for lower bounds on data structures, streaming algorithms, distributed computing, and the extension complexity of polytopes in~\cite{RaoY2020} and the survey~\cite{DBLP:conf/mfcs/Sherstov14}.

One somewhat natural concrete setting for the use of $L_n$ in information extraction might be the following: consider data in a CSV file with fixed columns from which we want to extract all pairs of lines that have identical entries in at least one column from a column set $S$. This can easily be modelled with the CFG formalisms proposed for information extraction in \cite{DBLP:conf/icdt/Peterfreund21,DBLP:conf/pods/AmarilliJMR22}, but if the algorithm requires unambiguous CFGs (as those in \cite{DBLP:conf/icdt/Peterfreund21,DBLP:conf/pods/AmarilliJMR22} do), then an easy reduction from $L_n$ shows that any such grammar must be of exponential size in the number of considered columns in~$S$. This lower bound remains true if instead of equality we require other natural comparison of the columns, say lexicographic order, similarity measures, and so on.

Thus, while $L_n$ is not itself natural from a data management perspective, we still feel that it might have practical applications. We will discuss these aspects in future versions of our paper.

\paragraph{Proof Outline}
As mentioned, we consider the  language $L_n$ from~\cite{DBLP:journals/corr/abs-2309-11663}. There, it was already shown that for some~$n$ the language $L_n$ admits a {\CFG} of size $\Theta(\log(n))$; for completeness we give such a grammar for all $n$ in Appendix~\ref{ap:smallCFG} (see also the discussion in Section~\ref{sec:preliminaries}). It was also remarked in~\cite{DBLP:journals/corr/abs-2309-11663} that~$L_n$ admits a nondeterministic finite automaton of size $\Theta(n)$; the idea is that the automaton first nondeterministically `guesses' the positions of the matching $a$ symbols and then verifies this guess. Therefore, to prove Theorem~\ref{thm:main} it suffices to prove the lower bound (3) for $L_n$. 

Thus far the best tools to analyse unambiguity seem to come from communication complexity, see~\cite{RaoY2020} for an introduction. Our paper is no different. We make a connection to communication complexity, motivated by work in knowledge compilation~\cite{BovaCMS16}, by relating the size of {\uCFG}s and so-called \emph{rectangles} (see Section~\ref{sec: rectangles}). To be precise, we show that if we cannot split up a language $L$ into a small number of disjoint rectangles, then $L$ does not admit a small {\uCFG}. 
Since rectangles have been studied extensively in communication complexity, one may hope to use results from that area as a `black box.' Unfortunately, this cannot be done since our setting is more general than the one normally studied in communication complexity. 
The main technical contribution of this paper is that we overcome this difference in settings. 

To do this we build on an approach from~\cite{DBLP:journals/siamcomp/Sherstov16}, by defining disjoint $A, B \subset L_n$, such that any rectangle contains roughly the same number of words from $A$ and $B$. On the other hand,~$L_n$ contains many more words from $A$ than $B$, so we need many disjoint rectangles to cover the whole language. This type of argument is known as a \emph{discrepancy argument} in the communication complexity literature, see e.g.~the textbook description~\cite[Chapter~5]{RaoY2020}.

\paragraph{Related Work}
It has long been known that for \emph{infinite} languages {\CFG}s are far more succinct than {\uCFG}s. More than 40 years ago, Schmidt and Szymanski~\cite{DBLP:journals/siamcomp/SchmidtS77} showed that in general the size of~a {\uCFG} representing a language $L$ cannot be bounded by any computable function in the size of~a {\CFG} for the same language $L$. 
However, this result is not true for \emph{finite} languages, which are more interesting in the database setting that we consider. 
There, our doubly exponential separation is optimal, in the sense that every {\CFG} accepting a finite language can be transformed into an equivalent {\uCFG} with at most a double-exponential blow-up \cite{DBLP:journals/corr/abs-2309-11663}.

This is not the first time that different representations of finite languages have been compared qualitatively. In fact, since the influential paper~\cite{DBLP:journals/tcs/BucherMCW81}, the representation of finite languages has been an active research field, see e.g. \cite{DBLP:journals/jalc/CampeanuH04, DBLP:conf/dcfs/HolzerW18, DBLP:journals/mst/CaselFGGS21}. 
Actually, in \cite{DBLP:journals/tcs/BucherMCW81} {\uCFG}s and {\CFG}s are also compared for succinctness. However, 
they measure the size of a grammar in terms of the total number of rules whereas we measure the sum of the sizes of all rules. 
Our size measure corresponds to the size of factorised representations while the measure in~\cite{DBLP:journals/tcs/BucherMCW81} is far smaller: a constant number of rules can create languages that need linear size factorised representations. Consequently, 
our results, are quite different. Perhaps the paper closest in flavour to ours is by Filmus~\cite{DBLP:journals/ipl/Filmus11}, where lower bounds on the sizes of {\CFG}s for various finite languages are shown, with respect to our size model.

Another influential research strand considering representation by {\CFG}s 
is that of \emph{grammar based compression} \cite{DBLP:journals/tit/KiefferY00, DBLP:journals/tit/KiefferYNC00, DBLP:journals/tit/YangK00}. There one aims to find a small {\CFG} representing a single word $w$; we may think of this process as compressing a long document.\footnote{Note, that in this context {\CFG}s are often called \emph{straight-line programs}.}
  Algorithms can then be applied directly to the compressed document  \cite{DBLP:journals/gcc/Lohrey12}. 
  This paradigm has become increasingly influential in database theory in recent years, see e.g.~\cite{DBLP:conf/pods/SchmidS21,DBLP:conf/pods/SchmidS22,DBLP:journals/pacmmod/LohreyS24,DBLP:conf/pods/SchmidS22a,DBLP:conf/icdt/MunozR23}. Since we are interested in {\CFG}s that represent many strings 
  our results are quite different from this line of work.

	\section{Preliminaries}\label{sec:preliminaries}

We write $[n]:= \{1, 2, \dots, n\}$ and for integers $i,j$ with $i\le j$, we write $[i,j]:= \{\ell\in \mathbb{Z}\mid i\le \ell \le j\}$. For~a set $U$ we write $\mathcal{P}(U)$ to denote its power set, i.e., the set containing every subset of $U$. Moreover,  let $U$ and $V$ be sets of sets, such that $A \cap B = \emptyset$ for all $A \in U$, $B \in V$. Then we may identify each pair $(A,B)$ with the set $A \cup B$ and so reflecting this we will write $U \times V : = \{A \cup B \, \mid \, A \in U, B \in V\}$.

We assume that the reader is familiar with some basics of context free grammars, see e.g.~\cite{Sipser97}. Let $\Sigma$ be a set of symbols which we call an \emph{alphabet}. A \emph{word} over $\sigma$ is a finite string of symbols from $\Sigma$. For a word $w$ we write $|w|$ to denote its length. We write $\Sigma^{\ast}$ for the set of all words over~$\Sigma$. Then a \emph{language over $\Sigma$} is defined to be a set $L \subseteq \Sigma^{\ast}$. This paper is only concerned with \emph{finite languages} which means that $L$ is a finite set. 

\begin{definition}
A \emph{context free grammar} ({\CFG}) is a four-tuple $G = (\Sigma, N, R, S)$ where 
\begin{itemize}
\item $\Sigma$ is a finite set of symbols called the \emph{terminals},
\item $N$ is a finite set called the \emph{non-terminals},
\item $R$ is a set of \emph{rules} of the form $A \to W$ where $A \in N$ and $W \in (\Sigma \cup N)^{\ast}$, and
\item $S \in N$ is the \emph{start symbol}.
\end{itemize} 
\end{definition}

When we have two rules $A\rightarrow W$ and $A\rightarrow W'$ with $W\ne W'$, we sometimes write them more compactly as $A\rightarrow W \mid W'$. This notation always has to be interpreted as two rules and not as a single, more complex rule. 

{\CFG}s generate a language over the alphabet $\Sigma$ as follows. Let $A \to W$ be a rule in $R$ and let $u,v \in \Sigma^{\ast}$. Then we say that $uAv$ \emph{yields} $uWv$, in symbols $uAv \rightarrow uWv$. For $W, W' \in (\Sigma \cup N)^{\ast}$, we say that $W$ derives $W'$ if there is a finite sequence $(w_i)_{i \in [k]}$ such that $w_1 = W$, $w_k = W'$ and $w_i \rightarrow w_{i+1}$ for all $i \in [k-1]$. We write $W \rightarrow^{\ast} W'$. Then the language of a grammar $G$ is defined as $L(G):= \{w \in \Sigma^{\ast} \, \mid \, S \rightarrow^{\ast} w\}$. We also say that the {\CFG} \emph{accepts} $L(G)$. We define the size of a {\CFG} to be $|G| := \sum_{(A \to W) \in R} |W|$. Note that---in contrast to infinite languages---every finite language is accepted by some context free grammar, so the concern of our work is not expressivity of {\CFG}s but succinctness of representation by {\CFG}s.

We now give an example that we will use for our main result.

\begin{example} \label{ex: ambig}
Define the language 
\[
L_n := \{ (a + b)^{k} a (a+b)^{n-1} a ( a + b)^{n-1-k} \, \mid \, k \le n-1\}
\]
consisting of words of length $2n$ that contain two $a$ symbols with a word of length exactly $n-1$ between them. In \cite{DBLP:journals/corr/abs-2309-11663} the following small {\CFG}s are defined which accept $L_{2^n+1}$.

Let $G_n$ be the grammar with terminals $\{a,b\}$, non-terminals $\{A_i, B_i\}_{0 \le i \le n}$, start symbol $A_n$ and rules,
\begin{align*}
A_i &\to B_{i-1}A_{i-1} \mid A_{i-1}B_{i-1} & \text{ for } i \in [n] \\
A_0 &\to B_0 a B_n a \mid a B_n a B_0 \\
B_i &\to B_{i-1}B_{i-1} & \text{ for } i \in [n] \\ 
B_0 &\to a \mid b.
\end{align*}
Then $G_n$ has size $\Theta(n)$ and accepts $L_{2^n +1}$ see \cite{DBLP:journals/corr/abs-2309-11663} for details. 
\end{example}

In fact the above grammar can be adapted to give a $\Theta(\log(n))$ sized {\CFG} accepting $L_n$ for every~$n$. For completeness we give such a grammar in Appendix~\ref{ap:smallCFG}.

\begin{figure}
    \begin{subfigure}{.4\textwidth}
    \begin{tikzpicture}[
        scale=.5
        ]
        \node{$A_1$}
        child {node {$A_0$}
            child {node {$B_0$}
                child {node {$a$}}
            }
            child {node {$a$}}
            child {node {$B_1$}
                child {node {$B_0$}
                    child {node {$a$}}
                }
                child {node {$B_0$}
    child {node {$a$}}
}
            }
            child {node {$a$}}
        }
        child[missing] {}
        child[missing] {}
        child {node {$B_0$} 
            child {node {$a$}}}
    ;
    \end{tikzpicture}
    \end{subfigure}
    \begin{subfigure}{.4\textwidth}
    \begin{tikzpicture}[
        scale=.5
        ]
        \node{$A_1$}
        child {node {$B_0$} 
            child {node {$a$}}}
        child[missing] {}
        child[missing] {}
        child {node {$A_0$}
            child {node {$a$}}
            child {node {$B_1$}
                child {node {$B_0$}
                    child {node {$a$}}
                }
                child {node {$B_0$}
                    child {node {$a$}}
                }
            }
            child {node {$a$}}
            child {node {$B_0$}
                child {node {$a$}}
            }
        }
        ;
    \end{tikzpicture}
       \Description{Two different parse trees for the word $aaaaaa$ for the grammar of Example~\ref{ex: ambig}.}
\end{subfigure}
    \caption{Two different parse trees for the word $aaaaaa$ for the grammar of Example~\ref{ex: ambig}.}
    \label{fig:ambiguous}

\end{figure}
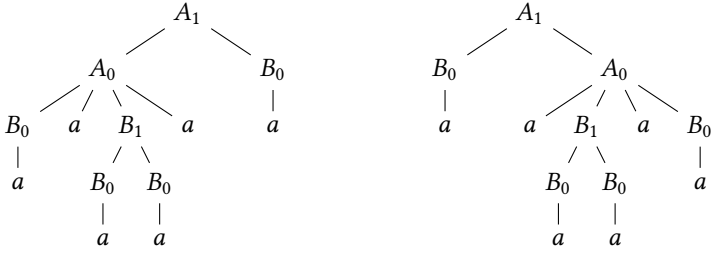

Each derivation in a context free-grammar may be associated with a tree, called a parse-tree, in the natural way (see Figure~\ref{fig:ambiguous}). We say that a {\CFG} is \emph{unambiguous} if every $w \in L(G)$ has a unique parse tree; equivalently this means that every word in $L(G)$ has a unique derivation in $G$. We call such a grammar a {\uCFG}.

\begin{example}
    The grammar from Example~\ref{ex: ambig} is not unambiguous; see Figure~\ref{fig:ambiguous} for two different parse trees of the same word. So let us construct an unambiguous grammar for $L_n$.
    
    In a first step, it will be useful to generate all possible words of a given length $i$. To this end, we introduce non-terminals $C_1, \ldots, C_n$ and rules 
    \begin{align*}
        C_i &\to a C_{i-1} \mid b C_{i-1}& \text{ for } i \in [n], i>1 \\
        C_1 &\to a \mid b.
    \end{align*}
    Clearly, any derivation from any $C_i$ is unambiguous. Next, for every word $w\in \Sigma^{\le n}$, we introduce a non-terminal $A_w$ and a rule
    \begin{align*}
        A_w &\to w.
    \end{align*}
    One way to enforce unambiguity for $L_n$ is to ensure that 
    each derivation of a word $w$ in $L_n$ determines the first pair of positions containing $a$ at distance $n$ in $w$. To this end, we introduce for every $i\in [n]$ a way to generate words $u$ such that $i$ is the smallest integer for which the letters at position $i$ and $i+n$ in $u$ are both $a$. For every $i<n$, we thus introduce a non-terminal $A_i$ and for every $w\in \Sigma^{i-1}$ the rule
    \begin{align*}
        A_i &\to A_w a C_{n-i} A_{\bar{w}}a C_{n-i}, 
    \end{align*}
    where $\bar{w}$ is the word we get from $w$ by complementing all its letters, i.e., by flipping all occurrences of $a$ to $b$ and vice-versa. For $i=n$, we introduce $A_n$ and the rules
    \begin{align}\label{eq:uCFGexample}
        A_n &\to A_w a A_{\bar w} a
    \end{align} 
    for all $w\in \Sigma^{n-1}$.
    Clearly, all derivations from any $A_i$ are unambiguous. It only remains to add a starting symbol $S$ and the rules
    \begin{align*}
        S &\to A_1 \mid \ldots \mid A_n.
    \end{align*}
    Note that the grammar we have just constructed accepts $L_n$ but it has exponential size in $n$; in particular, there are exponentially many rules in (\ref{eq:uCFGexample}). So the unambiguous grammar here is of double exponential size with respect to the ambiguous one from Example~\ref{ex: ambig}. Our main result will show that this is unavoidable because every grammar for $L_n$ has exponential size in $n$.
\end{example}

We say that a {\CFG} is in \emph{Chomsky normal form} if all rules are of the form $A\rightarrow BC$ or $A\rightarrow a$ where $A,B,C$ are non-terminals and $a$ is a terminal \cite{DBLP:journals/iandc/Chomsky59a}. It is well-known that any {\CFG} $G$ can be transformed into an equivalent one $G'$ in Chomsky normal form, such that $|G'| \le |G|^2$. Therefore, from now on we always assume that {\CFG}s are in Chomsky normal form. Moreover, we assume that we don't have any redundant non-terminals, i.e. that every non-terminal $A$ appears in at least one parse tree generated by the grammar. If this is not the case, we get a smaller grammar accepting the same language by deleting all rules containing $A$.

	\section{From Unambiguous CFGs to Rectangle Covers} \label{sec: rectangles}

In this section, we show how one can prove lower bounds for unambiguous {\CFG}s by means of communication complexity. In fact, we focus on {\CFG}s whose languages are finite and, in particular, whose words are all of the same length. The following definition is crucial.

\begin{definition}[Rectangle] \label{def: rect} 
We say that a language $L$ with words of length $n$ is a \emph{rectangle} if there are numbers $n_1, n_2, n_3$ such that:
\begin{align*}L = \bigcup_{\substack{w_1\!w_3\in L_1: \\ |w_1| = n_1, |w_3| = n_3}} 
    \{w_1\}\times L_2 \times \{w_3\} \text{ where }L_1 \subseteq \Sigma^{n_1+n_3}, L_2 \subseteq \Sigma^{n_2}.
\end{align*}
We also say that $L$ is a rectangle with \emph{parameters} $(L_1, L_2,n_1, n_2, n_3)$. We say that a rectangle is \emph{balanced} if and only if $\frac{n}{3} \le n_2 \le \frac{2n}{3}$.
\end{definition} 

\begin{example}
Consider the language
\[
L_n^{\ast} := a^{n/2} (a+b)^{n} a^{n/2}
\]
of all words of length exactly $2n$ which begin and end with $n/2$ consecutive $a$ symbols. Then $L_n^{\ast}$ is a balanced rectangle. This follows by setting $n_1 = n_3 = n/2$, $n_2 = n$, $L_1 = \{a^n\}$ and $L_2 = (a+ b)^n$. 
\end{example}

It is easy to see that every language can be written as a union of balanced rectangles since any language containing a single word is a balanced rectangle. The aim of this section is to show that the number of rectangles needed to cover a language is closely related to the size of {\CFG}s accepting the language. This is formalised by the following proposition.

\begin{proposition}\label{prop:rectangles}
    Let $L$ be a language in which all words have length $n$ and which is accepted by a {\CFG}~$G$. Then $L$ can be written as
    \begin{align}\label{eq:rectangles}
        L:= \bigcup_{i\in [\ell]} L_i
    \end{align}
    where every $L_i$ is a balanced rectangle and $\ell \le n|G| 
    $. Moreover, if $G$ is unambiguous, then the union in (\ref{eq:rectangles}) is disjoint.
\end{proposition}

With Proposition~\ref{prop:rectangles} in hand, to prove that a language does not admit a small {\uCFG} it suffices to show that whenever $L$ is equal to a disjoint union of rectangles, the number of rectangles in the union must be big.  This gives a lower bound technique analogous to those which have been successfully deployed in \emph{knowledge compilation}, see for example \cite{BovaCMS16}.

\begin{example}
$L_n$ can be written as a union of $n$ balanced rectangles. To see this observe that for all fixed $0 \le k \le n-1$ the language
\[
L_n^k := (a+b)^k a (a+b)^{n-1} a (a+b)^{n-1-k}
\]
is a balanced rectangle, setting $n_1 = k$, $n_2 = n+1$, $n_3 = n-1-k$, $L_1 = (a+b)^{n-1}$ and $L_2 = a(a+b)^{n-1}a$. Moreover $L_n$ is by definition the union of all these languages. However, note that the $L_n^k$ languages are not disjoint. In fact, in Section~\ref{sec: lb} we will see that any union of balanced rectangles which equals $L$	must have size $2^{\Omega(n)}$, which by Proposition~\ref{prop:rectangles} proves our main lower bound.
\end{example}

In the remainder of this section, we will prove Proposition~\ref{prop:rectangles}. So fix a language $L$ in which all words have the same length $n$ and a context free grammar $G$ accepting $L$. We start with the following simple observation.

\begin{observation}\label{obs:samelength}
    For every non-terminal $A$ of $G$, the words that can be generated by a derivation in~$G$ starting from $A$ all have the same length.
\end{observation}

\begin{proof}
    By way of contradiction, assume that this were not the case, so there is some $A$ from which we can generate words of different lengths. Then, by assumption, we know that $A$ appears in~a parse tree of $G$. But in that parse tree we can substitute the derivation below $A$ by derivations for words of different lengths, and we get that $L$ contains words of different lengths, which contradicts the assumption on $L$.
\end{proof}

With the help of Observation~\ref{obs:samelength}, we can rewrite $G$ into an equivalent grammar in the following way: for every non-terminal $A$ of $G$ introduce new non-terminals $A_1, \ldots A_n$. For every rule $A\rightarrow BC$ let $\ell$ be the length of the words generated by $B$ and $i \in [n]$. Then if $i+\ell\le n$ the grammar $G'$ has rules $A_i\rightarrow B_i C_{i+\ell}$. Moreover, for every rule $A\rightarrow a$, the grammar $G'$ has the rules $A_i\rightarrow a$ for $i\in [n]$. Finally, the start symbol of $G'$ is $S_1$.

\begin{lemma} \label{lem:equiv}
    The grammars $G$ and $G'$ generate the same language and $|G'| \le n|G|$. Moreover, if $G$ is unambiguous then so is $G'$.
\end{lemma}    
\begin{proof}
    We build on the observation that in $G'$ the index $i$ of the non-terminals $A_i$ determines the position of the first letter of the word generated from $A_i$ in every word generated by $G$.     
    We use this to show that there is a bijection between the parse trees for the grammar $G$ and $G'$. So fix~a parse tree for a word $w$ in~$G$. Let $A$ be a non-terminal symbol in this parse tree and say that $A$ generates the word $w'$. Let $i$ be the position of the first letter of $w'$ in $w$. Then we substitute $A$ by~$A_i$ in the parse tree. Doing this for every non-terminal yields a new tree which is easily seen to be a parse tree of $G'$ for $w$. The inverse construction is simply deleting the indices that we introduced in the construction of $G'$. Overall, this shows that $G$ and $G'$ generate the same languages. Moreover, since we have a bijection of parse trees, if $G$ is unambiguous then so is $G'$.
    
    For the size bound, observe that for every $i\in [n]$ and every rule of $G$, we introduced a new rule of the same size. Thus, the size grows by a factor of $n$.
\end{proof}

We now make a connection between non-terminals of $G'$ and rectangles.

\begin{observation} \label{ob: rect}
    Let $A_i$ be a non-terminal of $G'$. Let $L'$ consist of all the words in the language of~$G'$ with a parse tree containing $A_i$. Then $L'$ is a rectangle for some parameters $(L_1, L_2, n_1, n_2, n_3)$ where the words in $L_2$ are exactly those generated by $A_i$.
\end{observation} 
\begin{proof}
    Consider any word $w$ in $L'$. Then $w= w_1 w_2 w_3$, where $w_2$ is generated from $A_i$. Then we can exchange $w_2$ by any other word $\bar w_2$ generated from  $A_i$ to get $\bar w = w_1\bar w_2w_3$ which is also in $L'$, so $\{w_1\}\times L(A_i)\times \{w_3\}\subseteq L'$ where $L(A_i)$ consists of all words generated from $A_i$.
    
    By Observation~\ref{obs:samelength}, we know that all words in $L(A_i)$ have the same length, that $|w_1|= i-1$ and $|w_3|= n - (i-1) - |w_2|$. Therefore, for all $w' = w_1'w_2'w_3'\in L'$ with $w_2'\in L(A_i)$ we have $|w_1| = |w_1'|$, $|w_2| = |w_2'|$  and $|w_3| = |w_3'|$. So we can write
    \begin{align*}
        L'= \bigcup_{\substack{w_1\!w_2\!w_3\in L': \\ w_2\in L(A_i)}} \{w_1\}\times L(A_i) \times \{w_3\}
   \end{align*}  
    which proves the observation.
\end{proof}

We can now construct the set of rectangles in Proposition~\ref{prop:rectangles} iteratively via the following process. While the language of $G'$ is not empty, choose a word $w$ in it and a parse tree for $w$. We claim that this parse tree must contain a non-terminal $A_i$ that generates words of length between $\frac{n}{3}$ and $\frac{2n}{3}$. To find this non-terminal, we use the following standard procedure: we descend in the parse tree from the root, iteratively always following the edge to the child whose subtree contains more leaves. We stop when the current node $v$ has less than $\frac{2n}{3}$ leaves in its subtree and claim that we have found the desired non-terminal which is the label of $v$. To see this, note that, since we have not stopped before, the subtree of the parent of $v$ contains more than $\frac{2n}{3}$ leaves, and since we go to the child with more children, $v$ must have at least $\frac{n}{3}$ leaves in its subtree. Observing that the number of leaves below $v$ is the same as the length of the generated word proves the claim. 

We put~$L'$, defined as in Observation~\ref{ob: rect}, into the set of rectangles we are constructing. Then we delete $A_i$ from~$G'$ as well as all other non-terminals that no longer appear in any parse tree and continue the iteration with the resulting language. 

This process terminates after at most $|G'| \le n|G|$ iterations and so produces some set of rectangles $\{R_i\}_{i \in [\ell]}$ with $\ell  \le n|G|$. Observe that, by construction, $\bigcup_{i \in [\ell]} R_i = L$. Moreover, if $G$ is unambiguous then, by Lemma~\ref{lem:equiv}, $G'$ is also unambiguous. This means that every word in $L$ appears in a unique parse tree in $G'$ and so the union $\bigcup_{i \in [\ell]} R_i$ is disjoint. Therefore, this process really does produce a set of rectangles witnessing Proposition~\ref{prop:rectangles}.

	\section{A Lower Bound for Disjoint Rectangle Covers} \label{sec: lb}

In this section, we show our lower bound for representations by {\uCFG}s via an exponential lower bound for the language $L_n$ from Example~\ref{ex: ambig}. This will provide the missing lower bound of Theorem~\ref{thm:main} and thus prove our main result.
Recall, that $L_n$ is defined as follows:
\[
L_n := \{ (a + b)^{k} a (a+b)^{n-1} a ( a + b)^{n-k-1} \, \mid \, k \le n-1\}.
\]

\noindent We will prove the following, confirming a conjecture from~\cite{DBLP:journals/corr/abs-2309-11663}.
\begin{theorem}\label{thm:lower}
For every $n\in \mathbb{N}$, every ${\uCFG}$ accepting $L$ has size $2^{\Omega(n)}$.
\end{theorem}

In the rest of this section, we prove Theorem~\ref{thm:lower}.
Observe that every word in $L_n$ has length $2n$ and that $L_n$ consists of all those words such that there are two $a$ symbols with distance exactly $n$ between them.  To prove Theorem~\ref{thm:lower}, we will use the rectangle based method presented in Section~\ref{sec: rectangles}.

Note that we may assume that $n$ is sufficiently big by choosing the constant hidden in the $\Omega$-notation to be sufficiently small. To simplify some of our arguments, for now we only consider~$L_n$ for values of $n$ which are divisible by $4$; it will be easy to remove this condition at the end of this section. 

\subsection{A Set Perspective}

Before diving into our proof we first introduce a change of perspective. It will be helpful to think of each word over $\{a,b\}$ of length $2n$ as defining a pair of subsets of~$[n]$. Formally, we define a bijection which maps each $w = w_1 \dots w_{2n}$, where $w_i \in \{a,b\}$ for each $i$ to a pair $(X_w, Y_w) \subseteq \{x_1, \dots, x_n\} \times \{y_1, \dots, y_n\}$  where $\{x_1, \dots, x_n\}$ and $\{y_1, \dots, y_n\}$ are both sets of elements and $X_w$ contains every $x_i$ such that $w_i = a$ and $Y_w$ contains every $y_i$ such that $w_{i+n} = a$. We write $X: =  \{x_1, \dots, x_n\}$ and $Y:= \{y_1, \dots, y_n\}$. 
Sometimes it will be useful to refer to the elements of $X \cup Y$ in a unified way, hence we define 
\[
z_i := \begin{cases}
x_i & i \in [n] \\
y_{i-n} & i \in [n+1, 2n].
\end{cases}
\]

\noindent We write $Z[i,j] : = \{z_{\ell} \, \mid \, i \le \ell \le j\}$ and $Z: = Z[1,2n]$. We similarly define $X[i,j]$ and $Y[i,j]$. Observe that each interval also induces a partition of~$Z$ as follows.

\begin{definition}[Ordered Partition]
We say that a partition $(\Pi_0, \Pi_1)$ of $Z$ is \emph{induced by the interval $[i,j]$} if for some $\ell \in \{0,1\}$, we have $\Pi_{\ell} = Z[i,j]$. We call a partition which is induced by some interval an \emph{ordered} partition. Moreover, we call an ordered partition $(\Pi_0, \Pi_1)$ \emph{balanced} if 
\[
\frac{2n}{3} \le |\Pi_0|, |\Pi_1| \le \frac{4n}{3}
\]
\end{definition}

Having these dual notions of intervals and the partitions they induce on $Z$ will be useful in our proofs. We next translate Definition~\ref{def: rect} to fit the set perspective.

\begin{definition}[Set Rectangle] \label{def: set_rect}
Let $(\Pi_0, \Pi_1)$ be an ordered partition of $Z$. We say that $R \subseteq \mathcal{P}(Z)$ is an ordered $(\Pi_0, \Pi_1)$-\emph{set rectangle} if there are sets $S \subseteq \mathcal{P}(\Pi_0)$ and $T \subseteq \mathcal{P}(\Pi_1)$ such that $R= S \times T$.\footnote{Remember that to simplify notation we read $S\times T$ as $\{U\cup V\mid U\in S, V\in T\}$.} We say that $R$ is a \emph{balanced set-rectangle} if $(\Pi_0, \Pi_1)$ is balanced. Moreover, if $(\Pi_0, \Pi_1)$ is induced by $[i,j]$ we also call~$R$ an $[i,j]$-set rectangle. 
\end{definition}

The following lemma justifies Definition~\ref{def: set_rect} and therefore our use of the set perspective in the rest of the proof.

\begin{lemma}
Let $L$ be a language over the alphabet $\{a,b\}$ with words of length $2n$, such that $L$ is~a rectangle in the sense of Definition~\ref{def: rect} with parameters $(L_1, L_2, n_1, n_2, n_3)$. Then $\{(X_w, Y_w) \, \mid \, w \in L\}$ is a $[n_1+1, n_1+n_2]$-set rectangle. 

Conversely, for any language $L$ over the alphabet $\{a,b\}$ with words of length $2n$ such that $\{(X_w, Y_w) \, \mid \, w \in L\}$ is an $[i,j]$-set rectangle, $L$ is a rectangle with $n_1 = i-1$, $n_2 = j-i+1$ and $n_3 = 2n-j$. 
\end{lemma}

\begin{proof}
Let 
\[L = \bigcup_{\substack{w_1\!w_3\in L_1: \\ |w_1| = n_1, |w_3| = n_3}} \{w_1\}\times L_2 \times \{w_3\}\text{ where }L_1 \subseteq \Sigma^{n_1+n_3}, L_2 \subseteq \Sigma^{n_2}.\]
be a rectangle with parameters $(L_1, L_2, n_1, n_2, n_3)$. Define 
\begin{align*}
S &:= \{(X_w, Y_w) \, \mid \, w = w_1 b^{n_2} w_3 \textit{ for some } w_1\!w_3 \in L_1\} \subseteq Z[1,n_1] \cup Z[n_1 + n_2 +1, 2n] \textit{ and}  \\
 T &:= \{(X_w, Y_w) \, \mid \, w = b^{n_1}w_2b^{n_3} \textit{ for some } w_2 \in L_2\} \subseteq Z[n_1+1, n_1 +n_2].
\end{align*} 
By construction $U \cap V = \emptyset$ for all $U \in S$, $V \in T$, so it is easy to verify that $S \times T = \{U \cup V \,\mid \, U \in S, V \in T\} = \{(X_w, Y_w) \, \mid \, w \in L\}$ and that this is a $[n_1+1, n_1 +n_2]$ rectangle. The converse direction is similar.
\end{proof}

Slightly abusing notation, we will from now on write $L_n$ to mean $\{(X_w, Y_w) \, \mid \, w \in L_n\}$. Also from now on we refer to ordered set-rectangles simply as ordered rectangles; that is, we will now completely switch to the set perspective. Seen this way, the language $L_n$ consists exactly of those pairs $(X_w, Y_w)$ of sets for which there is an $i$ such that $x_i\in X_w$ and $y_i\in Y_w$ are in the respective sets. Going one step further, if we identify the sets $X_w, Y_w$ with the corresponding index sets, $L_n$ consists of intersecting pairs of sets, so $L_n$ is essentially the complement of the famous set disjointness problem, which is widely regarded as the flagship problem of communication complexity, see e.g.~the survey~\cite{DBLP:conf/mfcs/Sherstov14}.

From the above discussion along with Proposition~\ref{prop:rectangles}, it follows that to prove Theorem~\ref{thm:lower} it suffices to prove the following proposition.

\begin{proposition} \label{prop: lb}
If $R_1, \dots, R_{\ell}$ is a set of ordered balanced rectangles such that $R_i$ and $R_j$ are disjoint for every $i \neq j$ and such that 
\begin{align*}
\bigcup_{i \in [\ell]} R_i = L_n,
\end{align*}
then $\ell = 2^{\Omega(n)}$.
\end{proposition}

We call a set of rectangles as in Proposition~\ref{prop: lb} a \emph{disjoint rectangle cover} for $L_n$. We will prove Proposition~\ref{prop: lb} in the next two subsections.

\subsection{A Lower Bound for Restricted Rectangles}

Before showing Proposition~\ref{prop: lb} in full generality, it will be useful to first consider a special case in which all rectangles share a particular partition.

\begin{theorem} \label{thm: base}
Any disjoint rectangle cover of $L_n$ by $[1,n]$-rectangles has size  $2^{\Omega(n)}$. 
\end{theorem}

Theorem~\ref{thm: base} is an immediate consequence of the so-called rank bound from communication complexity pioneered in~\cite{DBLP:conf/stoc/MehlhornS82}, see e.g.~\cite[Chapter~2]{RaoY2020} for a textbook introduction.
However, to prove Proposition~\ref{prop: lb} we need to generalise Theorem~\ref{thm: base}, which unfortunately cannot directly be done using known results from communication complexity: in contrast to most of the literature, Proposition~\ref{prop: lb} is not in the standard setting in which the partitions for all rectangles are fixed. Rather, different rectangles may use different partitions. This setting is called \emph{multi-partition communication complexity}~\cite{DBLP:journals/iandc/DurisHJSS04} and is far less studied. In particular, there is only one known lower bound that explicitly works for disjoint rectangle covers~\cite{DBLP:journals/tcs/Sauerhoff03}, which we cannot directly apply here. 

As a first step to proving Proposition~\ref{prop: lb} we thus give a new proof of Theorem~\ref{thm: base}. This proof will be more involved then the usual approach, but it has the advantage that we can later generalise it to Proposition~\ref{prop: lb}. We use a discrepancy argument in the spirit of that given by Sherstov in the context of multiparty communication complexity\footnote{We warn the reader that despite superficially similar names \emph{multiparty} and \emph{multipartition} communication complexity are very different concepts!}~\cite{DBLP:journals/siamcomp/Sherstov16}. The idea is to  define disjoint $A, B \subset L_n$, such that (1) any $[1,n]$-rectangle contains roughly the same number of sets from $A$ as from $B$ and (2) $L_n$ contains many more words from $A$ than $B$. Therefore, the size of any disjoint cover for~$L_n$ by $[1,n]$-rectangles must be large.

Informally, we split $X \cup Y$ into intervals of size four and then define $\Lset$ to be those sets $(U,V)$ containing exactly one element from each interval. Then $A$ is defined to be the set of $(U,V) \in \Lset$ such that the number of $i$ with $x_i \in U$ and $y_i \in V$ is odd and $B := \Lset \setminus A$. We next give the formal definition.

So set $m:= n/4$; note that $m$ is an integer since by assumption $n$ is divisible by four. Then for $i \in [m]$ we define:
\begin{align*}
 I_i^{X} :=& X[4(i-1)+1, 4i]\\
 \Lset_i^X :=& \{ U \subset I_i^X \, \mid \, |U|=1\} \\
 \Lset^X :=& \{ U \subset X \, \mid \, |U \cap I_i| = 1 \textit{ for every } i \in [m]\}.
 \end{align*}
We analogously define $I_i^Y$, $\Lset_i^Y$ and $\Lset^Y$ exactly as above but with $Y$ playing the role of~$X$. Then we set $\Lset := \Lset^X \times \Lset^Y$ and $\Lset_i = \Lset_i^{X} \times \Lset_i^Y$. Since we often want to take a unified perspective we define for $i \in [m]$, $I_i := I_i^X$ and $I_{i+m} = I_i^{Y}$. We call every $I_j$ an \emph{interval}.

We next show that $\Lset$ is big and that $L_n$ contains many more words from $A$ than from $B$. Recall, that $A \subset \Lset$ contains those $(U,V) \in \Lset$ such that the number of $i$ with $x_i \in U$ and $y_i \in V$ is odd and $B := \Lset \setminus A$.

\begin{lemma} \label{lem: sizes}
The following holds:
\begin{enumerate}
\item $|\Lset| = 2^{4m} $. \label{eq:lsize}
\item $|A \cap L_n| - | B \cap L_n| = |A| - |B \cap L_n| > 2^{\frac{7m}{2}}$. \label{eq:ABsize}
\end{enumerate}
\end{lemma}
\begin{proof}

For \eqref{eq:lsize} observe that in total there are $2m$ intervals each of length four and that for every such interval exactly one element is chosen for every $(U,V)\in \Lset$, so $|\Lset| = 4^{2m} = 2^{4m}$. 

The equality in \eqref{eq:ABsize} is trivial since $A \subset L_n$. To prove the inequality it suffices to show that $|B \setminus L_n| = 12^m$ and that $|B|- |A|= 2^{3m}$. To see this assume the aforementioned equalities, then
\begin{align*}
|A| - | B \cap L_n| &= |A| - (|B| - |B \setminus L_n|) \\ 
&= |B \setminus L_n| - (|B| - |A|) \\
&= 12^m - 2^{3m}  > 2^{\frac{7m}{2}},
\end{align*} 
where in the last step we use that $n$ (and thus $m$) is sufficiently big.
 
To show that $|B \setminus L_n| = 12^m$, consider for each $i\in [m]$ the $16$ sets $(U_i,V_i)\in \Lset_i$. By the definition of $\Lset_i$, the sets $U_i$ and $V_i$ contain a single element, respectively, so $U_i=\{x_j\}$, $V_i=\{y_k\}$, for some appropriate $j,k$. There are exactly four combinations such that $j=k$ and $12$ with $j\ne k$. Combining the latter for all $i\in [m]$ gives exactly $B \setminus L_n$, so $|B \setminus L_n| = 12^{m}$.

Finally, $|B|- |A|= 2^{3m}$ is true since
 \begin{align*}
2^{3m} &= (12 - 4)^{m} = \sum_{i=0}^m \binom{m}{i} 12^{m-i}(-4)^{i} \\
&= \sum_{i=0}^{\lfloor m/2 \rfloor}\left(\binom{m}{2i} 12^{m-2i} 4^{2i}\right ) - \sum_{i=1}^{\lceil m/2 \rceil}\left(\binom{m}{2i-1} 12^{m-(2i-1)} 4^{2i-1}\right ) \\
&= |B| - |A|. \qedhere
\end{align*} 
\end{proof}

The key to proving Theorem~\ref{thm: base} is the following: in Lemma~\ref{lem: sizes}, we have seen that there are many more sets in $A \cap L_n$ than in $B \cap L_n$. However, we will next show that each $[1,n]$-rectangle cannot contain many more sets of $A$ than $B$. Therefore, any disjoint cover of $L_n$ by $[1,n]$-rectangles must by big.

\begin{lemma} \label{lem: simple_des}
Let $R = S \times T$ be a $[1,n]$-rectangle. Then 
\[ \big||R \cap A| - |R \cap B|\big| \le 2^{3m}. \]
\end{lemma}

\begin{proof}
The key idea is to express $\big||R \cap A| - |R \cap B|\big|$ in terms of an appropriate expectation calculation. To set this up we need some definitions. 

Define for each $i\in [m]$ random variables $X_i, Y_i$ that take a value uniformly at random from $\Lset_i^{X}$ and $\Lset_i^{Y}$ respectively. We let all these random variables be independent and define new random variables $X := \bigcup_{i=1}^m X_i$, $Y = \bigcup_{i =1}^m Y_i$. Note that $X$ has the uniform distribution over $\Lset^{X}$ and $Y$ the uniform distribution on $\Lset^{Y}$. Moreover, $X$ and $Y$ are independent. Let $Y'$ be an independent copy of~$Y$.

For $L \subseteq \mathcal{P}(Z)$ we write $\chi_L$ for the indicator function of $L$, i.e. the function $\chi_L \colon \mathcal{P}(Z)\to \{0,1\}$ that evaluates to one exactly on inputs contained in $L$. For $U \in \Lset^X$ and $V \in \Lset^{Y}$ let $\Int(U,V) := 1$ if there is some $i$ such that $x_i \in U$ and $y_i \in V$, otherwise $\Int(U,V) :=0$. Note that if $U \in \Lset^X$ and $V \in \Lset^Y$ then $(-1)^{\chi_A(U,V)} = (-1)^{\sum_{i=1}^m \Int(U_i, V_i)}$, where $U_i := U \cap I_i^{X}$ and $V_i := V \cap I_i^{Y}$ for $i \in [m]$. 
The following calculation borrows heavily from \cite[Lemma 5.9]{RaoY2020} which in turn is inspired by an approach pioneered in~\cite{DBLP:journals/siamcomp/Sherstov16}.
\begin{align*}
\left(\frac{|R \cap A| - |R \cap B|}{ 2^{4m} }\right)^2 =& \left( \E_{X,Y} \left[\chi_R(X,Y) \cdot (-1)^{\chi_{A}(X,Y)}\right] \right)^2 \\
=& \left(\E_{X,Y } \left[ \chi_S(X) \cdot \chi_T(Y) \cdot (-1)^{\chi_{A}(X,Y)}\right] \right)^2 \\
\le& \E_{X} \left[ \chi_{S}(X)^2 \left( \E_{Y} \left[ \chi_T(Y) \cdot  (-1)^{\chi_{A}(X,Y)}\right] \right)^2 \right] \\
\le&  \E_{X, Y, Y'} \left[\chi_T(Y) \cdot \chi_T(Y') \cdot (-1)^{\chi_{A}(X,Y) +\chi_{A}(X,Y')}\right]\\
\le& \E_{Y, Y' }  \left[ \left| \E_{X} \left[(-1)^{\chi_{A}(X,Y) +\chi_{A}(X,Y')}\right] \right| \right] \\
=& \E_{Y, Y'} \left[ \left| \E_{X }\left[(-1)^{\sum_{i=1}^m \Int(X_i, Y_i) +\Int(X_i,Y'_i)}\right] \right | \right]
\end{align*}

Now fix $Y, Y'$. We claim that if $Y \neq Y'$ the inner expectation must evaluate to zero. First note that the probability that $Y = Y'$ is $1/|\Lset^{Y}| = 2^{-2m}$. Therefore, the claim implies the result via a simple calculation. We now prove the claim. 

So fix $Y \neq Y'$ while $X$ remains random as before. Let $C$ be the event that $\sum_{i=1}^m \Int(X_i, Y_i) +\Int(X_i,Y'_i)$ is even and $C_i$ be the event that $\Int(X_i, Y_i) +\Int(X_i,Y'_i) = 1$. Then $C$ is true if and only if the number of $i \in [m]$ such that $C_i$ is true is even. Note that if $Y_i = Y_i'$ then $C_i$ is false. So suppose $Y_i \neq Y_i'$. Since all intervals are of size four and $|X_i| = |Y_i| =|Y'_i| = 1$, we have $\mathbb{P}(C_i)= \frac{1}{2}$. Let $\mathcal{I} := \{i \in [m] \, \mid \, Y_i \neq Y_i'\}$ and $\alpha := |\mathcal{I}|$. Then
\begin{align*}
\mathbb{P}(C) = \sum_{i=0}^{\lfloor \frac{\alpha}{2} \rfloor} \binom{\alpha}{2i} 2^{-2i} 2^{-(\alpha- 2i)} = 2^{-\alpha} \sum_{i=0}^{\lfloor \frac{\alpha}{2} \rfloor} \binom{\alpha}{2i} = 2^{-\alpha} 2^{\alpha-1} = \frac{1}{2},
\end{align*}
\noindent where we use the fact that the sum of $\binom{\alpha}{i}$ over all even values of $i$ is $2^{\alpha-1}$. The claim and therefore the result follows.
\end{proof}

From Lemma~\ref{lem: simple_des}, we can deduce Theorem~\ref{thm: base}. In detail, let $R_1, \dots , R_\ell$ be a disjoint set of $[1,n]$-rectangles whose union is $L_n$. Then 
\begin{align*}
2^{\frac{7m}{2}} &< |A \cap L_n| - | B \cap L_n| \\
 &= \sum_{i=1}^\ell |A \cap R_i| - | B \cap R_i| \\ 
&\le  \ell \cdot 2^{3m},  
\end{align*}
where the first inequality is Lemma~\ref{lem: sizes}\eqref{eq:ABsize}, the equality uses the fact the rectangles form a cover and are disjoint, and the second inequality follows by Lemma~\ref{lem: simple_des}. It follows that $\ell = 2^{\Omega(m)} = 2^{\Omega(n)}$ which completes the proof of Theorem~\ref{thm: base}. 

\subsection{Proof of Theorem~\ref{thm:lower}}

To prove Proposition~\ref{prop: lb}, and therefore Theorem~\ref{thm:lower}, it suffices to prove an analogue of Lemma~\ref{lem: simple_des} that applies to all balanced ordered rectangles. In fact we have already done much of the work. The main point is that for every balanced ordered partition $(\Pi_0, \Pi_1)$ and `most' values of $\ell$, we have that $x_{\ell}$ and $y_{\ell}$ are on different sides of the partition. We now formalise the above idea, beginning with the nice case where the partition splits $x_{\ell}$ and $y_{\ell}$ for every $\ell$. Here using symmetry we may deduce the following corollary to the proof of Lemma~\ref{lem: simple_des}. 

\begin{corollary} \label{cor: disc}
Let $R = S \times T$ be an $[i,j]$-rectangle such that $j-i = n-1$. Then 
\[ \big||R \cap A| - |R \cap B|\big| \le 2^{3m} \]
\end{corollary}

To move from this to our more general case, we first elucidate some properties of balanced ordered partitions. So fix such a $(\Pi_0, \Pi_1)$. Let $G$ be the set containing those $i \in [n]$ such that $x_i$ and~$y_i$ are in different parts of the partition, i.e.~$G$ contains those $i \in [n]$ such that exactly one of~$x_i$ and~$y_i$ lies in $\Pi_0$. Let $V_G$ contain all those $x_i$ and $y_i$ such that $i\in G$. Moreover, let $I_G$ contain all intervals $I_i$ such that $I_i \subseteq V_G$.

It will be convenient to assume that $(\Pi_0, \Pi_1)$ is well-behaved with respect to the intervals $I_{\ell}$ in the following sense: we call $(\Pi_0, \Pi_1)$ \emph{neat} if for every interval $I_{\ell}$ we have that $I_{\ell} \subseteq \Pi_0$ or $I_{\ell} \subseteq \Pi_1$. In the remainder, we can largely restrict ourselves to neat partitions due to the following result.

\begin{lemma}\label{lem:makeneat}
  Let $R$ be an ordered balanced $(\Pi_0, \Pi_1)$-rectangle. Then there there is a neat ordered balanced partition $(\Gamma_0, \Gamma_1)$ and a set of disjoint $(\Gamma_0, \Gamma_1)$-rectangles $R_1, R_2, \ldots, R_k$ with $k \le 256$ such that $R = \bigcup_{\ell \in [k]}R_{\ell}$.
\end{lemma}

\begin{proof}
Suppose w.l.o.g.\! that $|\Pi_0| \le |\Pi_1|$. There are at most two intervals $I_i, I_j$ which violate neatness, i.e.\! that contain elements from $\Pi_0$ and $\Pi_1$. Let $(\Gamma_0, \Gamma_1)$ be the partition of $Z$ that we get by putting all elements from $I_i$ and $I_j$ into $\Pi_0$. Then, by construction, $(\Gamma_0, \Gamma_1)$ is neat and ordered. Moreover, it is balanced, since $\Gamma_0$ gained at most the $8$ elements from $I_i\cup I_j$, so $|\Gamma_0| \le |\Pi_0| + 8 \le n + 8\le \frac{2n}{3}$, where for the last inequality we use that $n$ is sufficiently big. 

Now consider $\alpha\subseteq I_i\cup I_j$. Then let $R_\alpha$ contain exactly those sets $U$ from $R$ for which $U\cap (I_i \cup I_j) =\alpha$. Then clearly $R_{\alpha} \cap R_{\alpha'} = \emptyset$ for $\alpha \neq \alpha'$. Note that $R_\alpha$ is both a $(\Pi_0, \Pi_1)$-rectangle and a $(\Gamma_0, \Gamma_1)$-rectangle since it is fully determined on $I_i\cup I_j$ and thus on the difference between the two partitions. Moreover, we have
    \begin{align*}
        R = \bigcup_{\alpha\subseteq I_i\cup I_j} R_\alpha.
    \end{align*}
    Since $|I_i\cup I_j| = 8$ and thus the disjoint union ranges over $2^8=256$ sets, the claim follows.
\end{proof}

Neat ordered partitions have the following useful properties.

\begin{lemma} \label{lem:ordered_par}
Let $(\Pi_0, \Pi_1)$ be a neat, ordered and balanced partition such that $|\Pi_0| \le |\Pi_1|$. Then
\begin{enumerate}
\item $\Pi_0 \subseteq V_G$ and
\item $|\Pi_0| = |G|$. 
\end{enumerate}
\end{lemma}

\begin{proof}
Let $[i,j]$ be the interval which induces the partition. First, suppose $\Pi_0 = Z[i,j]$. Since $|\Pi_0| \le n$ it follows that for every two elements $z_k, z_\ell$ of $\Pi_0$ the indices have distance at most $n-1$, i.e., if $z_k, z_{\ell} \in \Pi_0$ then $|\ell-k| \le n-1$. Since the distance between $x_{\ell}$ and $y_{\ell}$ is $n$ for every $\ell$, it follows that $\Pi_0$ cannot contain both $x_\ell$ and $y_\ell$ and thus $\Pi_0 \subseteq V_G$. 

If instead $\Pi_0 = Z \setminus Z[i,j]$, then 
\begin{align*}
\Pi_0 = \{z_{\ell} \in Z \, \mid \, \ell < i \} \cup \{z_{\ell} \in Z \, \mid \, \ell > j\}.
\end{align*}
So since $|Z[i,j]| \ge n$, we know that $j-i \ge n-1$ and it is easy to see that no two elements of $\Pi_0$ are at distance $n$. So again $\Pi_0 \subseteq V_G$; we have shown (1). 

For (2), observe that for every $\ell \in G$ we must have either $x_{\ell}$ or $y_{\ell}$ in $\Pi_0$, so $|G| \le |\Pi_0|$. Moreover, from (1) we get that for every $\ell \in G$ the set $\Pi_0$ cannot contain $x_{\ell}$ and $y_{\ell}$, so $|\Pi_0| \le |G|$ and thus $|\Pi_0| = |G|$. 
\end{proof}

We can now deploy Corollary~\ref{cor: disc} and Lemma~\ref{lem:ordered_par} to prove a generalisation of Lemma~\ref{lem: simple_des}.

\begin{lemma} \label{lem: gen_disc}
Let $(\Pi_0, \Pi_1)$ be a neat, ordered and balanced partition and let $R = S \times T$ be a $(\Pi_0, \Pi_1)$-rectangle.
Then 
\[ \big||R \cap A| - |R \cap B|\big| \le 2^{10m/3}. \]
\end{lemma}

\begin{proof}
Suppose w.l.o.g.\! that $|\Pi_0| \le |\Pi_1|$. Note that, since the partition is balanced, we have $|\Pi_0| \ge 2n/3$.
Define $\Pi_1^{g} := \Pi_1 \cap V_G$ and $\Pi_1^{b} := \Pi_1 \setminus \Pi_1^{g}$. Observe that, via Lemma~\ref{lem:ordered_par}, we get $|\Pi_0| = |\Pi_1^g|$ and thus  $|\Pi_1^{b}| = 2n - |\Pi_0| - |\Pi_1^g| = 2n - 2|\Pi_0| \le 2n/3$.

Now let $\alpha \subseteq \Pi_1^{b}$ be a subset of the `bad elements', i.e.~the elements $x_\ell, y_\ell$ that both lie in $\Pi_1$. Let $T_{\alpha}: = \{U \in T \, \mid \, U \cap \Pi_1^b = \alpha\}$ and $T^{\alpha} = \{ U \setminus \alpha \, \mid \, U \in T_{\alpha}\}$. Define $R^{\alpha} = S \times T^{\alpha}$. By construction, every element of $S \times T^{\alpha}$ is a subset of $V_G$. By renaming elements we can view $S \times T^{\alpha}$ as a subset of $L_{|G|}$, since all elements in $\Pi_0$ are in $V_G$ by Lemma~\ref{lem:ordered_par}(1). 
Then for every $\alpha \in U$, since $|G| = |\Pi_0|$---by Lemma~\ref{lem:ordered_par}(2)---and by applying Corollary~\ref{cor: disc} to $R^{\alpha}$ with $m' := |G|/4$, we obtain that
\begin{align*} \big||R^{\alpha} \cap A| - |R^{\alpha} \cap B|\big| \le 2^{3m'} = 2^{\frac{3|G|}{4}}.\end{align*}
Here we use that our partition is neat and thus that $G$ is a union of $m'$ many intervals.

Moreover, because $(\Pi_0, \Pi_1)$ is neat, $\Pi_1^{b}$ is  
a union of intervals. Suppose there is some $\ell$ such 
that $I_{\ell} \subseteq \Pi_1^{b}$ and $|\alpha \cap 
I_{\ell}| \neq 1$. Then $\alpha \not \in \Lset$ and thus 
$T_{\alpha} \cap \Lset = \emptyset$. Let $U \subseteq 
\Lset$ be the set of $\alpha \subseteq \Pi_1^{b}$ such 
that $T_{\alpha} \cap \Lset \neq \emptyset$. Then for all 
$\alpha \in U$ we have that $|\alpha \cap I_\ell| = 1$, 
for every $\ell$ such that $I_{\ell} \subseteq \Pi_1^{b}$. 
We obtain that
\begin{align*}
|U| \le 4^{\frac{|\Pi_1^b|}{4}} = 4^{\frac{2n-2|G|}{4}} = 2^{n-|G|}.
\end{align*}
\noindent It follows that 
\begin{align*}
\big||R \cap A| - |R \cap B|\big| &\le \sum_{\alpha \in U} \big||R^{\alpha} \cap A| - |R^{\alpha} \cap B|\big| \\
&\le 2^{n - |G|} \cdot 2^{\frac{3|G|}{4}} =  2^{n-\frac{|G|}{4}} \le 2^{\frac{5n}{6}} = 2^{\frac{10m}{3}}
\end{align*}
as required, where we use that $|G| = |\Pi_0| \ge 2n/3$.
\end{proof}

Proposition~\ref{prop: lb} now follows by a routine calculation similar to what we have already seen.

\begin{proof}[Proof of Proposition~\ref{prop: lb}]
We assume first, as before, that $n$ is divisible by four, so we can apply all results from Section~\ref{sec: lb}.
Let $R_1, \dots, R_\ell$ be an ordered disjoint cover of $L_n$. With Lemma~\ref{lem:makeneat}, we can assume that the partitions for all rectangles are neat, since this only changes the number of rectangles by a constant factor. Then
\begin{align*}
2^{\frac{7m}{2}} &< |A \cap L_n| - |B \cap L_n| \\
 &= \sum_{i=1}^\ell |A \cap R_i| - |B \cap R_i| \\ 
&\le  \ell \cdot 2^{\frac{10m}{3}}. 
\end{align*}
Where the first equality is Lemma~\ref{lem: sizes}\eqref{eq:ABsize}, the second uses the fact the $\{R_i\}_{i \in [\ell]}$ is a disjoint cover of~$L_n$ and the inequality follows by Lemma~\ref{lem: gen_disc}. Since $\frac{7}{2} >\frac{10}{3}$, we obtain that  $\ell = 2^{\Omega(m)} = 2^{\Omega(n)}$.

Now assume that $n$ is not divisible by four, so $n=4t+a$ for some $a\in [3]$.
We call the elements of the form $x_{4t+b}, y_{4t+b}$, where $b \in [a]$ the \emph{spare elements}. For each $i\in [\ell]$, let $R_i'$ contain exactly those sets $U$ from $R_i$, that do not contain any spare elements. Then clearly $\bigcup_{i \in \ell} R_i' = L_{4t}$ and moreover, the union is disjoint. Suppose $R_i$ is a $(\Pi_0, \Pi_1)$-rectangle. Then for $i \in \{0,1\}$ define $\Pi_i'$ to be the set obtained from $\Pi_i$ by removing all the spare elements. It follows that $R_i'$ is also a $(\Pi_0', \Pi_1')$ rectangle. The only remaining problem is that $R_i'$ might not be balanced.
But, by essentially the same argument as that used in Lemma~\ref{lem:makeneat}, there is an ordered balanced partition $(\Gamma_0, \Gamma_1)$ and a set of disjoint $(\Gamma_0, \Gamma_1)$-rectangles $T_1, \dots, T_k$ with $k \le 2^6 =64$ such that $R_i' = \bigcup_{j \in [k]} T_j$. We thus obtain an ordered disjoint cover of $L_{4t}$ of size at most $64\ell$. It follows that $\ell = 2^{\Omega(4t)} = 2^{\Omega(n)}$. 
\end{proof}

Combining Proposition~\ref{prop: lb} and Proposition~\ref{prop:rectangles} directly yields Theorem~\ref{thm:lower} and thus Theorem~\ref{thm:main}.

	\section{Conclusions}\label{sec:conclusions}

We have shown an optimal double-exponential separation in terms of succinctness of {\CFG}s from their unambiguous variant. Notably, this separation holds for a natural language over a binary alphabet. To obtain our separation we introduced a lower bound technique based on rectangles (Proposition~\ref{prop:rectangles}), inspired by methods from knowledge compilation~\cite{BovaCMS16}. However, our argument has~a different flavour to these lower bounds and so opens up interesting prospects for future work.

In particular, our lower bound applies to an \emph{unambiguous} representation format. Dealing with such models poses a challenge and perhaps our techniques can be applied more widely to gain a better understanding of unambiguity. For example, in the context of \emph{document spanners} it has been shown that \emph{rigid} grammars can be made unambiguous in exponential time \cite[Theorem 5]{DBLP:conf/pods/AmarilliJMR22}; can we show that this is optimal? Closer to the topic of our paper, can we prove that complementation is hard for {\uCFG}s over finite languages? We should note that studying the complement operation for unambiguous representation formats can be challenging. The case of unambiguous finite automata \cite{DBLP:conf/icalp/Raskin18,DBLP:conf/icalp/GoosK022} and structured d-DNNF circuits \cite{DBLP:conf/ijcai/Vinall-Smeeth24} have only recently been solved, and even then the lower bounds are only quasi-polynomial. Moreover, the case of d-DNNF circuits has been open for over two decades. We hope that this paper marks a step towards a better understanding of unambiguity.

	\begin{acks}
The authors thank Florent Capelli for helpful discussions. The first author would also like to thank the Simons Institute Fall 2023 program `Logic and Algorithms in Database
Theory and AI'. Here, via discussions with Wim Martens, he learned about the conjecture from~\cite{DBLP:journals/corr/abs-2309-11663} that is solved in this paper.

The first author was partially supported by the \grantsponsor{ANR}{Agence nationale de la recherche (ANR)}{https://anr.fr/} project \grantnum{ANR}{EQUUS ANR-19-CE48-0019}. The second author was funded by the \grantsponsor{DFG}{Deutsche Forschungsgemeinschaft (DFG, German Research Foundation)}{https://www.dfg.de/}  project number \grantnum{DFG}{414325841}. 

\end{acks}

	\newpage

	\bibliography{../extras/ordered}
	
	\newpage
	\appendix

\section{\texorpdfstring{ A small {\CFG} for $L_n$}{A small CFG for Ln}} \label{ap:smallCFG}

Fix an integer $n \in \mathbb{N}$. The idea will be to construct a word $w$ of length $n-1$ and then add another word $aw'a$ where $w'$ also has length $n-1$ at some position into $w$.

First, it will be useful to create words of specific lengths. To this end, for all $i$ with $2^i < n$, we introduce a non-terminal $B_i$ that generates all possible words of length $2^i$ by the rules
\begin{align*}
    B_i &\to B_{i-1}B_{i-1}  & \text{ for } i \text{ with } 1 \le 2^i  < n\\
    B_0 &\to a \mid b
\end{align*}  

From the binary representation of $n-1$, we can compute a set $I= \{i_1, \ldots, i_\ell\}$ of integers of size $O(\log(n))$ such that $n-1 = \sum_{i\in I}2^i$. Then we can imagine that $w$ is decomposed into blocks of length $2^i$ for $i\in I$, that is $w \in \{a,b\}^{2^{i_1}} \times \ldots \times \{a,b\}^{2^{i_\ell}}$. Creating the blocks of $w$ is easy with the $B_i$, but we still have to choose in which block we want to add $aw'a$. To this end, we introduce a binary tree $T=(V,E)$ whose leaves are bijectively labelled with the elements from $I$. For every node $v$ of $T$ we introduce two non-terminals $C_v, D_v$ where intuitively in a parse tree we will use $C_v$ if $aw'a$ is added in a subword corresponding to the blocks of the leaf labels below $v$ and $D_v$ otherwise. Concretely, for $v$ with children $u,w$, we add the rules
\begin{align*}
    C_v & \to C_u D_w \mid D_u C_w\\
    D_v & \to D_uD_w.
\end{align*}

For the leaves $v$, there are two cases: first we add for all leaves $v$ and all $i_j$ such that the label of $v$ is $i_j$  the rule
\begin{align*}
    D_v & \to B_{i_j} 
\end{align*} 
to generate all words of length $2^{i_j}$ from symbol $D_v$ (since there we do not want to add $aw'a$). 

For the $C_v$, we want to create all words of length $2^{i_j}$ and insert $aw'a$ at some position. We proceed similarly to Example~\ref{ex: ambig}. For all leaves $v$ and all $i_j$ such that the label of $v$ is $i_j$ we add the rule
\begin{align*}
    C_v & \to A_{I_j}.
\end{align*}
Then, we introduce symbols $A_i$ for all $i$ with $2^i < n$ and rules
\begin{align*}
    A_i & \to B_{i-1}A_{{i-1}} & \text{ for $i$ with } 2^i < n\\ 
    A_0 & \to B_0 a S a \mid a S a B_0,
\end{align*}
where $S$ is a new symbol from which we create $w'$ by the rule
\begin{align*}
    S & \to B_{i_1}\ldots B_{i_\ell}.
\end{align*}

Taking $C_r$, where $r$ is the root of $T$, as the starting symbol completes the construction of the grammar. By the discussion above, it generates $L_n$. So it only remains to show that the size of the grammar is $O(\log(n))$. 

First, observe that all rules but the last---which is of size $|I|=O(\log(n))$---have constant size. So it suffices to bound the number of rules. Noting further, that for every non-terminal $T$ we have only a constant number of rules $T\to W$, we can in fact just bound the number of non-terminals. By construction, we have $O(\log(n))$ non-terminals $A_i$ and $B_i$. Moerover, we have $O(|V(T)|) = O(|I|) = O(\log(n))$ non-terminals $C_v$ and $D_v$. The only additional non-terminal is the single $S$, which completes the proof.

\end{document}